\documentclass[journal,11pt,draftcls,onecolumn]{IEEEtran}

\IEEEoverridecommandlockouts

\usepackage{amsthm}
\usepackage{amssymb}
\usepackage{amsmath}
\usepackage{amssymb}
\usepackage{epsfig}
\usepackage{epsf}
\usepackage{subfigure}
\usepackage{graphicx}
\usepackage{cite}
\usepackage{url}
\usepackage[]{authblk}
\usepackage{color}
\usepackage{upgreek}
\usepackage{comment}

\usepackage{wrapfig}
\usepackage{bbm}

\newcommand\smallO{
  \mathchoice
    {{\scriptstyle\mathcal{O}}}
    {{\scriptstyle\mathcal{O}}}
    {{\scriptscriptstyle\mathcal{O}}}
    {\scalebox{.7}{$\scriptscriptstyle\mathcal{O}$}}
  }

\usepackage{pifont}

\def\BibTeX{{\rm B\kern-.05em{\sc i\kern-.025em b}\kern-.08em
    T\kern-.1667em\lower.7ex\hbox{E}\kern-.125emX}}

\newtheorem{theorem}{Theorem}

\newtheorem{lemma}{Lemma}

\newtheorem{remark}{Remark}

\newcommand{\msf}{\mathsf}
\newcommand{\lbp}{\left\{}
\newcommand{\rbp}{\right\}}
\newcommand{\lp}{\left(}
\newcommand{\rp}{\right)}

\newcommand{\BEC}{BPEC}
\newcommand{\A}{\mathrm{A}}
\newcommand{\B}{\mathrm{B}}
\newcommand{\T}{\mathrm{T}}
\newcommand{\name}{multi-modal \BEC}
\newcommand{\Name}{Multi-modal \BEC}
\newcommand{\na}{n_{\A}}
\newcommand{\nb}{n_{\B}}
\newcommand{\nt}{n_{\T}}
\newcommand{\da}{\delta_{\A}}
\newcommand{\db}{\delta_{\B}}
\newcommand{\dt}{\delta_{\T}}
\newcommand{\ba}{\beta_{\A}}
\newcommand{\bb}{\beta_{\B}}
\newcommand{\bt}{\beta_{\T}}
\newcommand{\bm}{\beta_{\mathrm{max}}}
\newcommand{\bmin}{\beta_{\mathrm{min}}}
\newcommand{\si}{\mathrm{SI}}

\begin{document}




\title{Capacity Results for Non-Ergodic Multi-Modal Broadcast Channels with Controllable Statistics}


\author{
Alireza~Vahid,~\IEEEmembership{Senior~Member,~IEEE},~and~Shih-Chun~Lin,~\IEEEmembership{Senior~Member,~IEEE}
\thanks{Preliminary results of this work will appear in IEEE International Symposium on Information Theory (ISIT) 2024~\cite{Bimodal-ISIT}.}
\thanks{Alireza Vahid is with the Electrical and Microelectronic Engineering Department at Rochester Institute of Technology, Rochester, NY 14623, USA. Email: {\sffamily alireza.vahid@rit.edu}.}  
\thanks{Shih-Chun Lin is with the Department of Electrical Engineering, National Taiwan University, Taipei, Taiwan. Emails: {\sffamily sclin2@ntu.edu.tw}.}
}

\maketitle


\begin{abstract}
Movable antennas and reconfigurable intelligent surfaces enable a new paradigm in which channel statistics can be controlled and altered. Further, the known trajectory and operation protocol of communication satellites results in networks with predictable statistics. The predictability of future changes results in a non-ergodic model for which the fundamentals are largely unknown. We consider the canonical two-user broadcast erasure channel in which channel statistics vary at a priori known points. We consider a multi-modal setting with two non-transient modes (whose lengths scale linearly with the blocklength) and an arbitrary number of transient modes. We provide a new set of outer-bounds on the capacity region of this problem when the encoder has access to causal ACK/NACK feedback. The outer-bounds reveal the significant role of the non-transient mode with higher erasure probability both on the outer and the inner bounds. We show the outer-bounds are achievable in non-trivial regimes, characterizing the capacity region for a wide range of parameters. We also discuss the regimes where the inner and outer bounds diverge and analyze the gap between the two. A key finding of this work is the significant gain of inter-modal coding over the separate treating of individual modes.
\end{abstract}

\begin{IEEEkeywords}
Multi-modal channels, non-ergodic fading, controllable statistics, feedback communications, channel morphing.
\end{IEEEkeywords}


\section{Introduction}
\label{Section:Introduction}

The ever-increasing wireless data demand is pushing communication in beyond 5G networks to higher and higher frequency bands where the abundance of bandwidth could potentially support much higher data rates, but where communication is also hindered by high path-loss and frequent blockage and interruptions~\cite{DebbahTHz22,taherkhani2020performance,ThzNoise,peng2018statistical}.
To alleviate these challenges, several new ideas and technologies have emerged.
For instance, movable antenna (MA)~\cite{EldarMA,zhu2023movable} and reconfigurable intelligent surface (RIS)~\cite{DiRenzo-2022-RIS-metasurface} enable channel morphing and provide the transmitter(s) with the opportunity to control channel statistics to better meet the desired objectives such as enhanced coverage, security, or multi-user scaling~\cite{zheng2020joint,nassirpour2023beamforming,zou2022scisrs,hoang2023secrecy,van2023enhancing,hoang2024physical}. 
The fact that channel statistics may be controlled and altered, brings upon a new theoretical paradigm of non-ergodic channels.
Traditionally in information theory, ergodic channels refer to those in which channel statistics are governed by independently and identically distributed processes~\cite{el2011network}.
Channel morphing is not the only reason to consider non-ergodic networks.
In low earth orbit (LEO) satellite communications, the trajectory of satellites and the resulting statistical variations can be well predicted~\cite{OutageLEO,LEO_OFDM_CE}.
Thus, a fundamental understanding of the capacity region of non-ergodic networks with controllable and/or predictable statistics is of great relevance.

The frequent interruptions of the communication links in higher bands can be modeled by erasure links, while the non-ergodic nature of MA/RIS-aided channels and LEO satellite communications can be captured by the a priori knowledge of statistical changes. 
Thus, to model these unique characteristics and shed light on the capacity of non-ergodic multi-user channels in higher bands, we introduce and study the multi-modal broadcast packet erasure channels (\BEC s). 
In the \name, during a communication block of length $n$, the channel may undergo several ``modes,'' where loosely speaking, channel statistics remain unchanged during a mode. 
We distinguish two classes of modes, the first is transient, which is motivated by the physical changes in MAs~\cite{zhu2024historical} or the delay in reconfiguring RISs~\cite{nassirpour2023beamforming}, and last for a combined length scaling as $\smallO(n)$. A non-transient mode on the other hand has a $\mathcal{O}(n)$ length and represents stable channel conditions.

More specifically, we consider a \name~with two receiver terminals and assume the channel experiences two distinct non-transient modes, namely modes $\A$ and $\B$, and a transient mode in between, namely mode $\T$, over a block-length of $n$.
During mode $\A$ with length $\na$, erasure probabilities are governed by independently and identically distributed (i.i.d.) across space and time Bernoulli $(1-\da)$ processes, whereas during modes $\T$ with length $\nt$ and $\B$ with length $\nb = n - \na - \nt$, they are governed by i.i.d. Bernoulli $(1-\dt)$ and $(1-\db)$ processes, respectively.
We assume the erasure probabilities and the lengths of the non-transient modes are known a priori and globally. 
This latter assumption is justified as for instance an MA or an RIS is controlled by the transmitter and the changes may be communicated to other nodes.
We finally assume the availability of common ACK/NACK signaling (\emph{i.e.}, delayed channel state feedback) from each receiver to the other nodes. 
Our contributions are then multi-fold.

We present a new set of outer-bounds on the capacity region of the non-ergodic \name~with feedback. 
The overall outer-bound region is the intersection of three individual regions: the most obvious region is that of the same problem but with instantaneous knowledge of the channels; the other regions rely on quantifying an extremal entropy inequality, which quantifies the amount of information leakage to an unintended user in our non-ergodic, multi-modal problem with feedback. 
The two regions play different roles. 
One relies on the minimum information leakage limit, while the other assumes the stronger non-transient mode is dominant throughout communications and quantifies the penalty incurred by making such an assumption. 
We show that \emph{transient} modes do not affect the asymptotic inner and outer bounds as long as the total length of these modes scales as $\smallO(n)$.

Interestingly, for the outer-bound region that relies on the minimum information leakage, the slope of the boundaries is governed by the non-transient mode with the weaker forward links. 
We show that this region is achievable for a wide set of parameters. 
The other region also presents an important notion, which is the penalty incurred when assuming the non-transient mode with the stronger forward links takes up most of the communication block. 
When this assumption is indeed true, the outer-bounds are close to the inner-bounds, but they quickly diverge as the assumption no longer holds. 

We then present an opportunistic inter-modal coding strategy, \emph{i.e.}, coding across the two non-transient modes, that achieves a strictly larger rate region compared to separate treatment of these modes, \emph{i.e.}, intra-modal coding. 
We show that for a wide range of parameters, the inner and outer bounds match, thus characterizing the capacity region in such cases.
The key idea in the inter-modal coding is to shift the multi-cast phase of network coding to the mode with stronger channels.
This is intuitively beneficial as multi-cast satisfies multiple users at the same time and if it can be done at a higher rate, then the overall achievable region would be larger.
When the relative lengths of the modes are such that it is not possible to strike the perfect balance, the inner and outer bounds deviate.
We discuss different variations of the strategy and explore potential improvements.

\noindent \textbf{Related work:} Shannon feedback does not increase the capacity of discrete memoryless point-to-point channels~\cite{shannon1956zero}, and provides bounded gains in the multiple-access channel~\cite{gaarder1975capacity,ozarow1984capacity}.  
On the other hand, feedback can provide significant gains when it comes to broadcast~\cite{maddah2012completely,VahidMaddah-Ali_16} and interference~\cite{suh2011feedback,AlirezaBFICDelayed,vahid2016two,vahid2018throughput,vahid2021topological} channels.
In the context of \BEC s, feedback capacity has been an active area of research. 
Most results consider the case in which all receivers provide ACK/NACK feedback to the transmitter~\cite{georgiadis2009broadcast,Wang_erasure_2,Tassiulas_erasure_K}; while the study of \BEC~with one-sided~\cite{sc2016ISIT,he2017two,lin2021capacity,chu2023broadcast} and partial/intermittent~\cite{dueck1980partial,VenkataramananPradhan_13,sclin2018IT,shayevitz2012capacity,IFB-ITW,vahid2021erasure} feedback has revealed the usefulness of imperfect feedback. 
Further, \BEC s with receiver cache have been studied in~\cite{bidokhti2016erasure,vahid2021content}. Perhaps, the most relevant result to this paper is the study of \BEC s in~\cite{heindlmaier2018capacity}, which assumes the mode at each time instant is independently and identically drawn from some underlying distribution, and thus corresponds to an ergodic channel. We will further discuss this related work in Remark~\ref{remark:Bidokhti}.

\noindent \textbf{Organization:} We present the \name~model in Section~\ref{Section:Problem_BIC}, and the main findings and insights of this work in Section~\ref{Section:Main_BIC}. The converse and achievability proofs are presented in Sections~\ref{Section:Converse_BiModal} and~\ref{Section:Achievability-bimodal}, respectively. Section~\ref{Section:Conclusion_BiModal} concludes the paper and discusses future steps.


\section{Problem Formulation}
\label{Section:Problem_BIC}

We revisit the classical two-user broadcast packet erasure channel (\BEC) with predictable statistical variations and (Shannon) feedback. We note that: (1) in \BEC s, the causal feedback essentially boils down to channel state information (CSI) feedback as the transmitter already knows the transmit signal and there is no other noise beyond channel erasure, and (2) the predictability of future variations in channel statistics would result in a non-ergodic setting, which differentiates this work from prior results.  


\noindent \underline{\bf \Name:} We assume communication happens over a blocklength of $n \in \mathbb{Z}^+$. We start with a couple of definitions first. A ``mode'' is a period of time in which the channel statistics (\emph{i.e.}, erasure probabilities of the forward links) remain unchanged. Then, a ``transient mode'' is a mode whose length scales as $\smallO(n)$, where we follow the standard Landau small $\smallO$ notation. As we will show in our results, the transient modes do not affect the asymptotic inner and outer bounds and thus, we may relax the assumption that the channel statistics must remain unchanged within a transient mode without affecting the results of this paper. Figure~\ref{Fig:MultiModal-BEC} depicts a two-user \name~with two non-transient modes, namely $\A$ and $\B$, and one transient mode, namely $\T$.  Further, in this work, we limit our attention to a specific subset of \name~with no more than \emph{two non-transient modes} and one transient mode in between. The results, as we will discuss later, can be easily extended to the case in which we have multiple transient modes as long as the total length of these transient modes is $\smallO(n)$. 

\begin{figure}[!ht]
\centering
\includegraphics[width = 0.35\columnwidth]{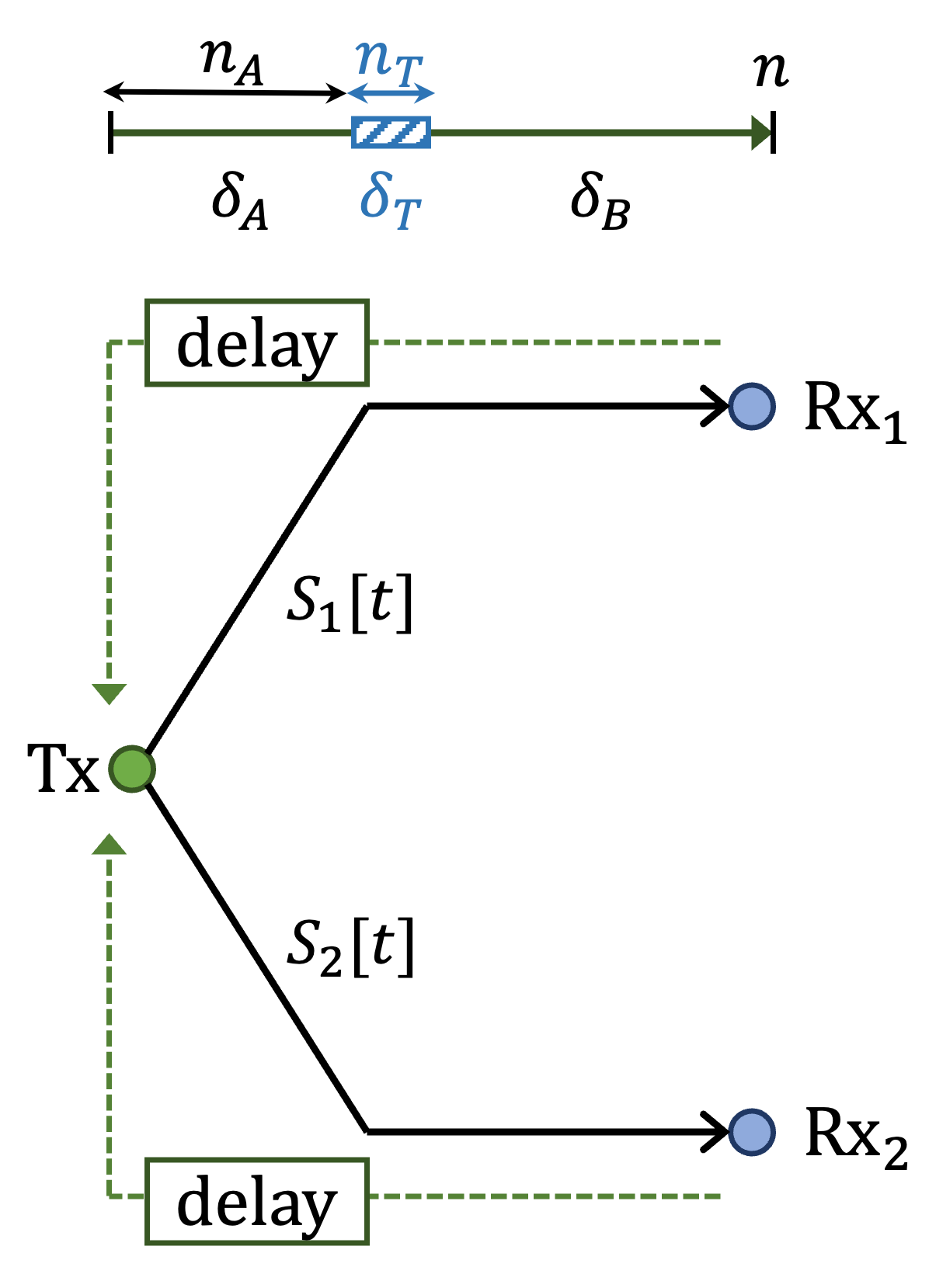}
\caption{Two-user \name~with (Shannon) feedback. We assume two non-transient modes and one (or more) transient modes with a total of length of $\smallO(n)$. 
\label{Fig:MultiModal-BEC}}
\end{figure}

\noindent \underline{\bf Channel statistics:} We assume during the first non-transient mode, the forward channel links, $S_i[t], i =1,2$, are governed by i.i.d. (over time and across users) Bernoulli $(1-\da)$ processes, and that this first mode terminates at some a priori known time, $1\leq n_{\A} \in \mathbb{Z}^+ \leq n$. Note that we assume $\na$ is $\mathcal{O}(n)$. During the transient mode, the forward channel links are governed by i.i.d. Bernoulli $(1-\dt)$ processes. Finally, during the second non-transient mode, the forward channel links are governed by i.i.d. Bernoulli $(1-\db)$ processes.

\begin{remark}[Orderly channels and ergodicity]
    In information theory, we traditionally assume two concepts in communication channels. First is the orderly nature of communications, which means if a symbol is transmitted before another, then its corresponding output will be delivered first. Second is ergodicity in the sense that the underlying processes that govern channel statistics are temporally independent. Recently, the orderly nature of communication channels have been revisited thanks to new paradigms such as macromolecular data storage systems~\cite{shomorony2021torn,outoforder}. The current work considers non-ergodic channels in which future statistical changes are known a priori.
\end{remark}


\noindent \underline{\bf Messages:} We follow the standard model where the transmitter, $\msf{Tx}$, wishes to transmit two independent messages (files), $W_1$ and $W_2$, to two receiving terminals $\msf{Rx}_1$ and $\msf{Rx}_2$, respectively, over $n$ channel uses. We note that no common message is included in this work. 
Each message, $W_i$, contains $|W_i|=m_i=nR_i$ data packets (or bits) where $R_i$ is the rate for user $i$, $i=1,2$. 
For simplicity and when convenient, we also denote message $W_1$ and $W_2$ as bit vectors $\vec{a} = \left( a_1,a_2,\ldots,a_{m_1} \right)$ and $\vec{b} = \left(b_1,b_2,\ldots,b_{m_2} \right)$, respectively. 
Here, we note that each packet is a collection of encoded bits, however, for simplicity and without loss of generality, we assume each packet is in the binary field, and we refer to them as bits. 
Extensions to broadcast packet erasure channels where packets are in large finite fields are straightforward as done in~\cite{ghorbel2016content,vahid2014communication}.

\noindent \underline{\bf Input-output relationship:} At time instant $t$, $1 \leq t \leq n$, the messages are mapped to channel input $X[t] \in \mathbb{F}_2$, and the corresponding received signals at $\msf{Rx}_1$ and $\msf{Rx}_2$ are given by:
\begin{align}
\label{eq_DL_channel}
Y_1[t] = S_1[t] X[t]~~ \; \mbox{and} \;~~ Y_2[t] = S_2[t] X[t],
\end{align}
respectively, where $\lbp S_i[t]\rbp$ denote the forward channels described above and are known to the receivers. When $S_i[t]=1$, $\mathsf{Rx}_i$ receives $X[t]$ noiselessly; and when $S_i[t]=0$, the receiver understands an erasure has occurred\footnote{This is an alternative way of describing the channel output as opposed assuming $\{ 0, \epsilon, 1\}$ as the possible output set.}. 

\noindent \underline{\bf CSI feedback and other assumptions:} We assume the receivers are aware of the instantaneous channel state information (\emph{i.e.} global CSIR). For the transmitter, on the other hand, we assume it learns the channel realizations of the forward links with unit delay (\emph{i.e.} delayed CSIT) but is aware of the erasure probabilities for only the non-transient modes. We further assume the duration of the first non-transient mode $\mathrm{n_A}$ is known globally. The knowledge of erasure probabilities, $\na$, and $\nt$ is considered as side-information denoted by $\si$. This side-information can be naturally obtained for MA/RIS. Furthermore, in LEO communication, the satellite trajectory would be known a priori, which is very helpful to predict the channel statistics and $\na$. We note that the transmitter only needs a subset of the $\si$ to achieve the results. For the outer-bounds, we enhance the channel by providing the entire $\si$ to the transmitter, but for achievability, we only rely on the knowledge corresponding to non-transient modes at the transmitter.

\begin{remark}[Feedback delay]
As mentioned earlier, a forward data packet is typically a long sequence of bits (typically $1000$s) , whereas the CSI feedback is essentially a short (one bit or a few bits at most) ACK/NACK signal. This observation justifies the global CSIR and the delayed CSIT assumptions above.
\end{remark}

\noindent \underline{\bf Encoding:} The delayed CSIT assumption and the availability of side information impose certain constraints on the encoding function. More specifically, at time index $t$, the encoding function $f_t(.)$ is described as:
\begin{align}
\label{eq_enc_function}
X[t] = f_t\lp W_1, W_2, S^{t-1}, \si \rp,
\end{align}
for $S^{t-1}=(S_1^{t-1}, S_2^{t-1})$. For simplicity and to avoid introducing additional notation, we included $\si$ in the encoding function, but as we will show in Section~\ref{Section:Achievability-bimodal}, the transmitter only relies on the knowledge related to non-transient modes.
 
\noindent \underline{\bf Decoding:} Each receiver $\msf{Rx}_i$, $i=1,2$ based on the global CSIR assumption knows the CSI across the entire transmission block, $S^n$. Then, the decoding function is $\varphi_{i,n}\left( Y_i^n, S^n, \sf{SI} \right)$. An error occurs whenever $\widehat{W}_i \neq W_i$. The average probability of error is given by:
\begin{align}
\lambda_{i,n} = \mathbb{E}[P(\widehat{W}_i \neq W_i)],
\end{align}
where the expectation is taken with respect to the random choice of the transmitted messages.

\noindent \underline{\bf Capacity region:} We say that a rate-pair $(R_1,R_2)$ is achievable, if there exist a block encoder at the transmitter and a block decoder at each receiver, such that $\lambda_{i,n}$ goes to zero as the block length $n$ goes to infinity. The capacity region, $\mathcal{C}$, is the closure of the set of all achievable rate-pairs. 


\section{Main Results \& Insights}
\label{Section:Main_BIC}

In this section, we present our main findings, and provide further insights and intuitions about the results. First, we introduce some notations based on which we define a set of regions to facilitate the expression of the main results. Then, we provide a set of outer bounds on the capacity region of the \name~with CSI feedback. Next, we provide an achievability scheme for the two-user \name~that achieves the outer-bounds under certain scenarios and otherwise outperforms the scheme that treats the two non-transient modes separately.

\subsection{Parameters \& Definitions}

We define the average erasure probability, $\bar{\delta}$, as:
\begin{align}
    \label{Eq:averagedelta}
    \bar{\delta} \overset{\triangle}= \lim_{n \rightarrow \infty}\frac{\na \da + \nt \dt + (n-\na-\nt ) \db}{n} = \eta \da + (1-\eta) \db,
\end{align}
for $\eta \overset{\triangle}= \lim_{n \rightarrow \infty} \na/n$, and where we used the fact that $\nt$ is $\smallO(n)$.
We further define: 
\begin{align}
    \ba \overset{\triangle}= 1 + \da, \quad \bb \overset{\triangle}= 1 + \db, \quad \bar{\beta} \overset{\triangle}= 1 + \bar{\delta}. 
\end{align}
We note that $\ba$ and $\bb$ would have defined the slope of the boundaries of the capacity region of \BEC s with CSI feedback where forward channels where governed only by $(1-\da)$ or $(1-\db)$, respectively. Moreover, if we were looking at a problem where the forward channels were governed by i.i.d. Bernoulli$(1-\bar{\delta})$, then $\bar{\beta}$ would have played the same role. We set:
\begin{align}
    \label{Eq:BetaMax}
    \bm \overset{\triangle}= \max\{ \ba, \bb \}, \qquad \bmin \overset{\triangle}= \min\{ \ba, \bb \}.
\end{align}

Finally, we define the following set of regions:
\begin{equation}
\label{Eq:Region-Outer-da}
\mathcal{C}_{1} \equiv 
\left\{ \begin{array}{ll}
0 \leq \bm R_1 + R_2 \leq \bm \left( 1 - \bar{\delta} \right), & \\
0 \leq  R_1 + \bm R_2 \leq \bm \left( 1 - \bar{\delta} \right), &  
\end{array} \right.
\end{equation}
and
\begin{equation}
\label{Eq:Region-Outer-db}
\mathcal{C}_{2} \equiv 
\left\{ \begin{array}{ll}
0 \leq R_i \leq \left( 1 - \bar{\delta} \right), & \\
\bmin R_1 + R_2 \leq \bmin \left( 1 - \bar{\delta} \right) + \kappa, & \\
R_1 + \bmin R_2 \leq \bmin \left( 1 - \bar{\delta} \right) + \kappa, &  
\end{array} \right.
\end{equation}
where
\begin{align}
\label{Eq:kappa}
\kappa = \mathbbm{1}_{\max\{\da,\db\} = \da} \eta/\ba (1-\da^2) + \mathbbm{1}_{\max\{\da,\db\} = \db} (1 - \eta)/\bb  (1 - \db^2), 
\end{align}
and finally,
\begin{equation}
\label{Eq:Region-Outer-full}
\mathcal{C}_{3} \equiv 
\left\{ \begin{array}{ll}
0 \leq R_i \leq \left( 1 - \bar{\delta} \right), & \\
0 \leq R_1 + R_2 \leq \eta (1-\da^2)+ (1 -\eta )(1-\db^2). &  
\end{array} \right.
\end{equation}


\subsection{Statement of the Main Results}

The first result establishes an outer-bound on the capacity region of our problem. 

\begin{theorem}
\label{THM:Bimodal-Outer}
For the two-user \name~with statistical variations and CSI feedback as described in Section~\ref{Section:Problem_BIC}, we have:
\begin{equation}
\label{Eq:Bimodal-Outer}
\mathcal{C} \subseteq \mathcal{C}_{1} \cap \mathcal{C}_{2} \cap \mathcal{C}_{3},
\end{equation}
where $\mathcal{C}_{1}, \mathcal{C}_{2},$ and $\mathcal{C}_{3}$ are defined in \eqref{Eq:Region-Outer-da}, \eqref{Eq:Region-Outer-db}, and \eqref{Eq:Region-Outer-full}, respectively.
\end{theorem}

\begin{remark}[Comparison of Theorem~\ref{THM:Bimodal-Outer} to the preliminary work~\cite{Bimodal-ISIT}]
Compared to~\cite{Bimodal-ISIT}, Theorem~\ref{THM:Bimodal-Outer} presents several improvements. The first improvement is the inclusion of non-transient modes in the channel model and quantifying the corresponding impact on the outer-bounds. The second improvement is the inclusion of $\mathcal{C}_{2}$, which ensures the capacity region matches that of a uni-modal problem in either extreme case $\eta = 0$ or $1$ as expected.
\end{remark}

To provide further insights on the outer-bounds, we first look at the three regions that are used in Theorem~\ref{THM:Bimodal-Outer}. $\mathcal{C}_{3}$ is the capacity region with \emph{instantaneous} feedback, which automatically serves as an outer-bound for the causal-feedback scenario we consider. To interpret $\mathcal{C}_{1}$ and $\mathcal{C}_{2}$, we recall that if instead of having a multi-modal structure, the forward channels were instead governed by Bernoulli$(1-\bar{\delta})$, with $\bar{\delta}$ being the average of $\da$ and $\db$, \emph{i.e.}, an ergodic setting, then $\bar{\beta}$ would have defined the slope of the boundaries of the region. Now, for $\mathcal{C}_{2}$, the slope of the boundaries is dominated instead by the \emph{smaller} of the two erasure probabilities, \emph{i.e.}, one that results in $\bmin$; while for $\mathcal{C}_{1}$, the slope of the boundaries is dominated instead by the larger of the two erasure probabilities, \emph{i.e.}, one that results in $\bm$. 

The more interesting region is indeed $\mathcal{C}_{1}$, which as we will see later, can be achieved for a wide range of parameters. On the other hand, $\mathcal{C}_{2}$, is included to ensure that on the extreme case where the length of the mode with the smaller erasure probability is close to $n$, we achieve the capacity as expected. This latter outer-bound region includes a correction term $\kappa$ given in \eqref{Eq:kappa} to compensate the preference given to the mode with lower erasure probability, and quickly deviates from the achievability strategy proposed later in Section~\ref{Section:Achievability-bimodal}. In both regions, the corner points come from the average probability of each link being active, \emph{i.e.} $(1-\bar{\delta})$.

\begin{remark}[Room for improvement]
\label{remark:improvement}
One may identify a weakness in the presented outer-bounds $\mathcal{C}_{1}$ and $\mathcal{C}_{2}$, which is the fact that the slope of the bounds is oblivious to the relative lengths of the non-transient modes. 
This does necessarily mean the bounds are loose and in fact, we will show that in non-trivial cases, the outer-bound region is indeed the capacity region of the problem.
However, based on the point highlighted above, the bounds could potentially be improved in other cases.
\end{remark}

The following theorem shows there exist non-trivial instances in which the outer-bound region in Theorem~\ref{THM:Bimodal-Outer} is achievable and matches the capacity region. The conditions presented in Theorem~\ref{THM:Bimodal-Outer} are of practical relevance: the erasure probability during the first mode is high and thus the MA/RIS is adjusted to improve the connectivity during the second mode.

\begin{theorem}
\label{THM:Bimodal-Inner}
For the two-user \name~with statistical variations and CSI feedback as described in Section~\ref{Section:Problem_BIC}, the outer-bounds in \eqref{Eq:Bimodal-Outer} are achievable when the transmitter has side-information only related to non-transient modes and:
\begin{align}
\label{Eq:Conditions}
\da \geq \db, \qquad \text{~and~} \qquad \eta \geq \left( 1 + \frac{\da(1-\da)}{2(1-\db)} \right)^{-1}.
\end{align}
\end{theorem}

\begin{remark}[Comparison of Theorem~\ref{THM:Bimodal-Inner} to the preliminary work~\cite{Bimodal-ISIT}]
In~\cite{Bimodal-ISIT}, we simply provided one example to prove that there exist $\da, \db,$ and $\eta$ such that $\da \neq \db$, $(1-\da)(1-\db) > 0$, and $0 < \eta < 1$ for which the outer-bounds in \eqref{Eq:Bimodal-Outer} are achievable. Theorem~\ref{THM:Bimodal-Inner} enhances this result by extending the results to a wide range of parameters. We further analyze the scheme when the conditions above are not met and examine the gap to the outer-bounds.
\end{remark}

To prove Theorem~\ref{THM:Bimodal-Inner}, it suffices to show one example meeting the corresponding conditions for which the outer-bound region is achievable. 
The general idea in achieving the capacity region of the two-user \BEC~with feedback is to send packets intended for each user in separate phases and then recycle these packets into ``multi-cast'' packets that upon delivery would simultaneously benefit both users~\cite{georgiadis2009broadcast}. 
For the \name, the key idea is to push the multi-cast phase as much as possible to the mode with lower erasure probabilities to boost the achievable rates.
This way the network coding spans across different modes, which we refer to as inter-modal coding.
We show this approach strictly outperforms the strategy the treats the modes separately.

\begin{remark}[Comparison to prior work]
\label{remark:Bidokhti}
In \cite[Sec IV-C]{heindlmaier2018capacity}, a finite-state memoryless \BEC~was studied where the erasure probabilities of $S_i[t]$ were specified through the conditional distribution given the current state. In other words, there is another channel state that controls the time-varying erasure statistics of the \BEC. Note that our model can also be similarly described by choosing a two-state non-random state sequence where the state changes at fixed time $\na$. However, in~\cite{heindlmaier2018capacity}, the capacity region is found only when the random state sequence is independently and identically (i.i.d.) generated as \cite{TimeVaryWang}, which is fundamentally different from our setting where the state sequence is fixed. 
Finally, though in~\cite[Proposition 3]{wang2014MIMOerasurecapacity} the i.i.d. Markovian state assumption could be removed, the capacity region was not reported, but instead the linear network coding rate region was found where the encoder in~\eqref{eq_enc_function} is limited to linear functions.
\end{remark}

The proofs for Theorem~\ref{THM:Bimodal-Outer} and Theorem~\ref{THM:Bimodal-Inner} are presented in Sections~\ref{Section:Converse_BiModal} and~\ref{Section:Achievability-bimodal}, respectively.

\begin{figure}[!ht]
\centering
\includegraphics[width = .5\columnwidth]{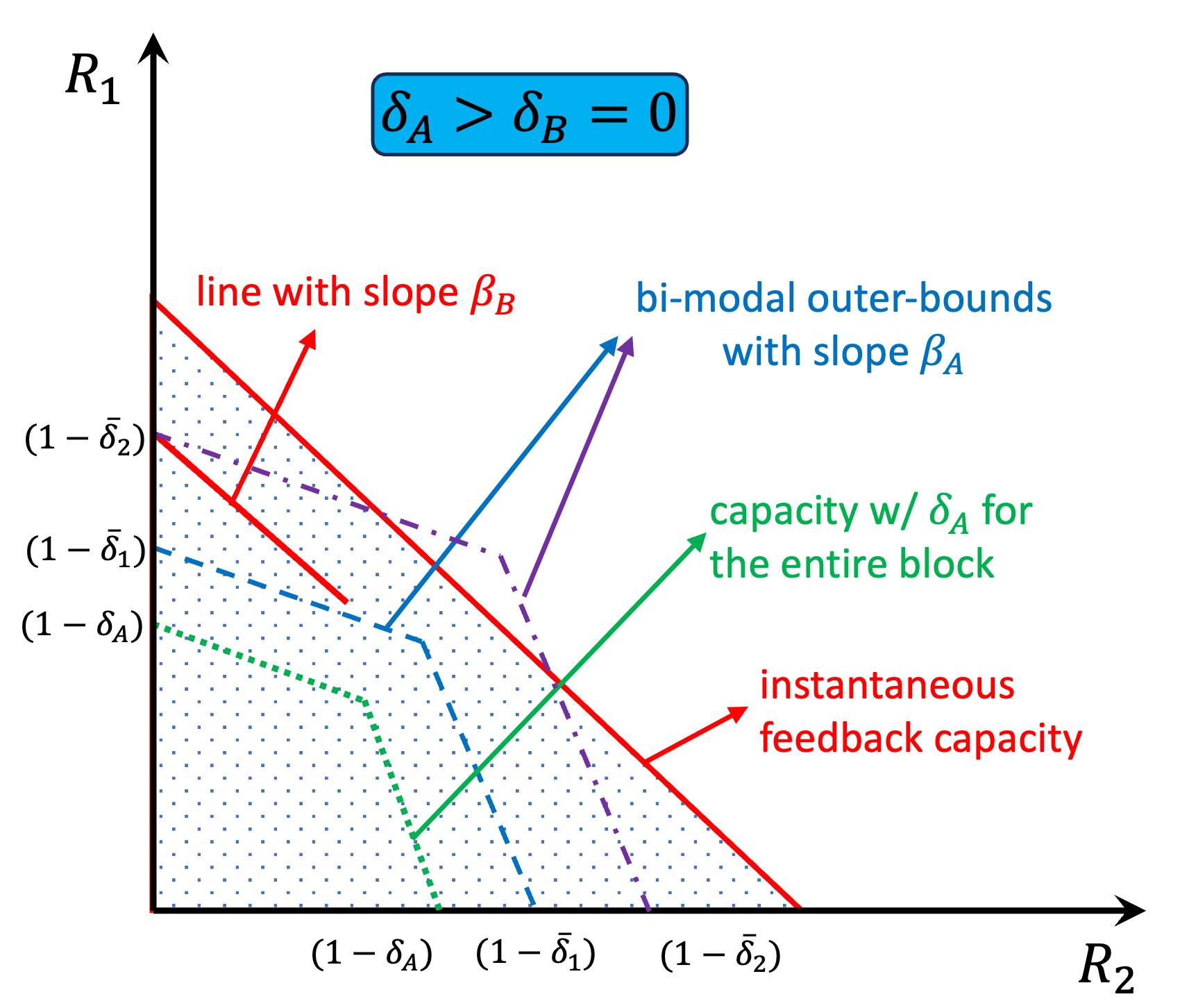}
\caption{Illustration of the outer-bound region of Theorem~\ref{THM:Bimodal-Outer} for $\da > \db = 0$ and various values of $\na$.\label{Fig:Region-BiModal}}
\end{figure}

\subsection{Illustration of the results}
\label{Section:Illustration}

In this sub-section, we illustrate the results and provide further insights. For simplicity and to keep the comparison meaningful to our preliminary work~\cite{Bimodal-ISIT}, in this part we consider a \name~with $\da > \db = 0$, which means during the second mode no erasures occur and that $\bm = \ba > \bb = 1$. 
Then, if $\na = 0$, then the capacity region is described by:
\begin{align}
R_1 + R_2 \leq 1,
\end{align}
which matches the instantaneous feedback capacity (as no erasures would occur); if $\na = n$, then the capacity region is described by:
\begin{equation}
\label{Eq:Bimodal-onlymodeA}
\left\{ \begin{array}{ll}
0 \leq \ba R_1 + R_2 \leq \ba \left( 1 - \da \right), & \\
0 \leq  R_1 + \ba R_2 \leq \ba \left( 1 - \da \right), & 
\end{array} \right.
\end{equation}
which is the capacity region of a \BEC~with homogeneous erasure probability of $\da$ for the entire communication block. 
Figure~\ref{Fig:Region-BiModal} (borrowed from~\cite{Bimodal-ISIT}) includes these two baselines as well as the outer-bound regions for two cases where $0<\na<n$. 
As depicted, the slope of the outer-bounds is dominated by $\ba$ as opposed to $\bar{\beta}$, and for this particular choices of $\da$ and $\db$, the outer-bound region given in~\eqref{Eq:Region-Outer-db} does not become active.
As the length of the second non-transient mode increases, meaning that $\na$ decreases, the instantaneous feedback outer-bound becomes active. 

\begin{figure}[!ht]
\centering
\includegraphics[width = .6\columnwidth]{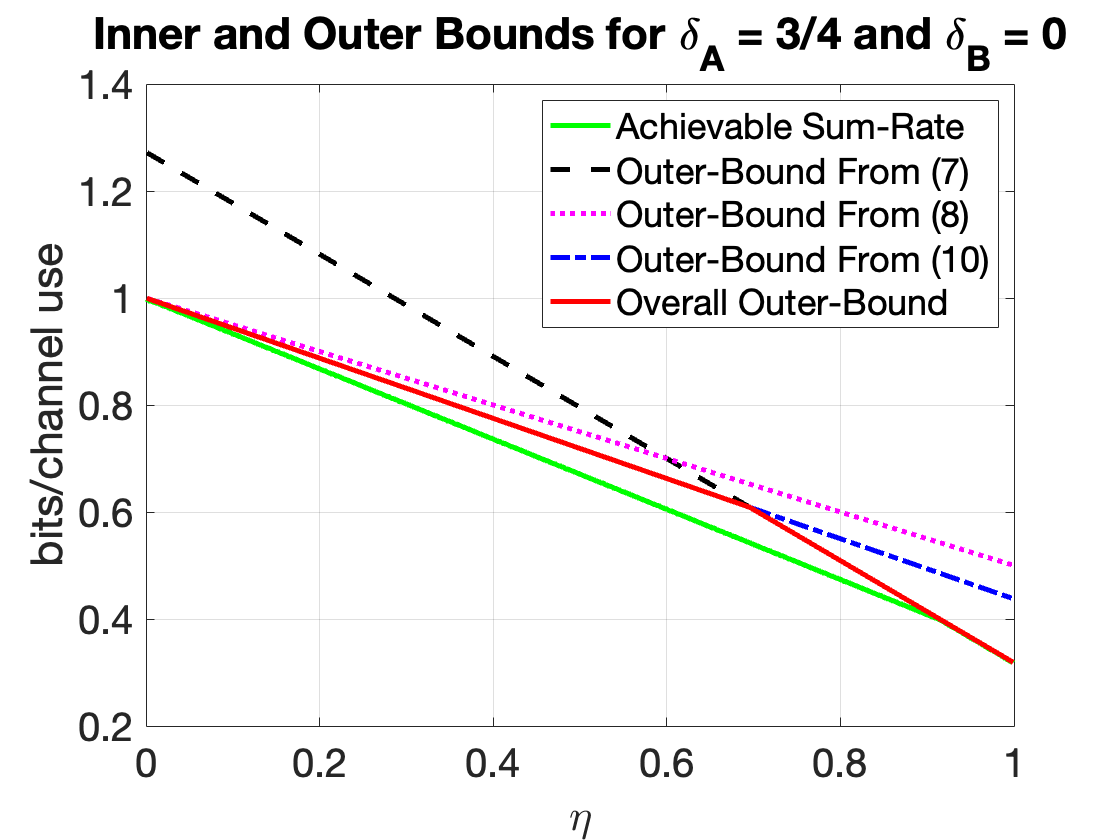}
\caption{Symmetric sum-rate inner-bound vs. outer-bounds for $\da = 0.75$ and $\db = 0$.\label{Fig:SumRate-BiModal}}
\end{figure}

While Figure~\ref{Fig:Region-BiModal} shows the entire region, Figure~\ref{Fig:SumRate-BiModal} focuses on the sum-rate outer-bound of the problem as implied by Theorem~\ref{THM:Bimodal-Outer}.
As we can see, depending on the value of $\eta$, different bounds may be active. 
For small values of $\eta$, the instantaneous outer-bound is active and as $\eta$ becomes larger, the new outer-bound region, $\mathcal{C}_1$, presented in~\eqref{Eq:Region-Outer-da} becomes active. 
We note that $\eta \in \{0, 1 \}$ corresponds to a homogeneous setting for which as expected the capacity is achievable. 
The non-trivial region is for $32/35 < \eta < 1$  for which the inter-modal achievability strategy is described in Section~\ref{Section:Achievability-bimodal} and is shown to achieve the outer-bounds thereby characterizing the capacity.

\begin{figure}[!ht]
\centering
\includegraphics[width = .6\columnwidth]{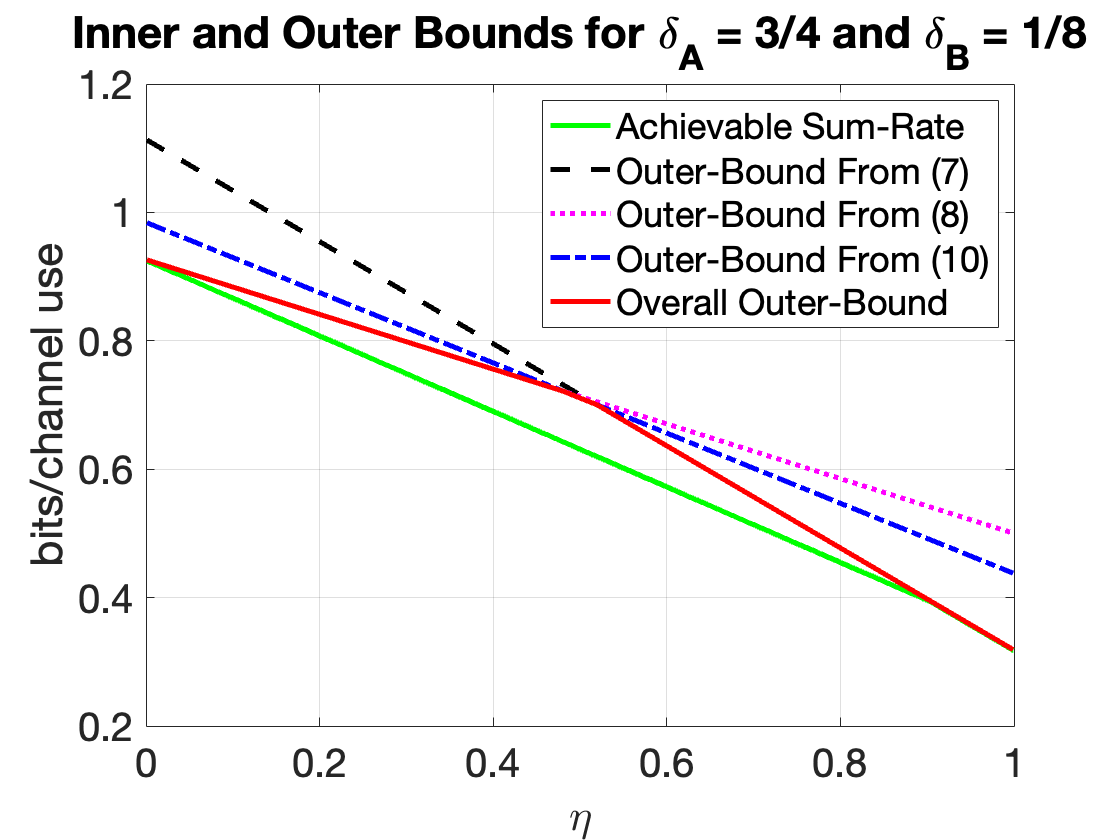}
\caption{Symmetric sum-rate inner-bound vs. outer-bounds for $\da = 0.75$ and $\db = 1/8$.\label{Fig:SumRate-BiModal-enhanced}}
\end{figure}

As we assumed $\db = 0$, the example presented above may appear somewhat contrived as feedback in the second non-transient mode is essentially useless. 
Moreover, for that example as seen in Figure~\ref{Fig:SumRate-BiModal}, the sum-rate outer-bound from \eqref{Eq:Region-Outer-db} is never active.
Thus, we instead look at an example where $\da = 3/4$ and $\db = 1/8$.
For this example, each one of the three regions defined in \eqref{Eq:Region-Outer-da}--\eqref{Eq:Region-Outer-full} contribute to the overall outer-bound at some value of $\eta$ as shown in Figure~\ref{Fig:SumRate-BiModal-enhanced}.

\section{Converse Proof of Theorem~\ref{THM:Bimodal-Outer}}
\label{Section:Converse_BiModal}

In this section, we derive the three outer-bound regions that contribute to the overall outer-bounds in Theorem~\ref{THM:Bimodal-Outer}. 
For the results in this part, we enhance the channel by providing the entire side-information to the transmitter (\emph{i.e.}, including any information related to the transient modes).
Any outer-bound on the capacity region of this enhanced channel will naturally serves as an outer-bound on our problem.
Later and in Section~\ref{Section:Achievability-bimodal}, we show that the transmitter does not need know $\nt$ or $\dt$ to achieve the promised rates.

\noindent \underline{$\mathbf{\mathcal{C}_{3}}$:} As mentioned earlier, this region corresponds to a network with instantaneous CSI feedback. For the first inequality in \eqref{Eq:Region-Outer-full} and for $i=1,2$, we have:
\begin{align}
n& R_i \leq H(W_i) \nonumber \\
&\overset{\mathrm{Fano}}\leq I(W_i; Y_i^n|S^n,\si) + n \upxi_n \nonumber \\
&=  H(Y_i^n|S^n,\si) - H(Y_i^n|W_i, S^n,\si) + n \upxi_n \nonumber \\
&\leq  H(Y_i^n|S^n,\si) + n \upxi_n \nonumber \\
&\leq {\na}(1-\da) + \nt (1-\dt) + (n-\na-\nt)(1-\db) + n \upxi_n,
\end{align}
and diving both sides by $n$ and taking the limit as $n \rightarrow \infty$ gives the desired result since $\nt$ scales as $\smallO(n)$.

The second inequality in \eqref{Eq:Region-Outer-full} is derived as follows:
\begin{align}
n&(R_1 + R_2) \leq H(W_1,W_2) \nonumber \\
&\overset{\mathrm{Fano}}\leq I(W_1, W_2; Y_1^n, Y_2^n|S^n,\si) + n \upxi_n \nonumber \\
&=  H(Y_1^n, Y_2^n|S^n,\si) - \underbrace{H(Y_1^n, Y_2^n|W_1, W_2, S^n,\si)}_{=~0} + n \upxi_n \nonumber \\
&\leq {\na}(1-\da^2) + \nt (1-\dt^2) + (n-\na-\nt)(1-\db^2) + n \upxi_n,
\end{align}
once again, diving both sides by $n$ and taking the limit as $n \rightarrow \infty$ gives the desired result since $\nt$ scales as $\smallO(n)$.

\noindent \underline{$\mathbf{\mathcal{C}_{1}}$:} Below, we derive the following outer-bound and the other bound is immediately derived by interchanging user IDs:
\begin{align}
\label{Eq:ConverseBound}
0 \leq R_1 + \bm R_2 \leq \bm \left( 1 - \bar{\delta} \right).
\end{align}

Suppose rate-tuple $\lp R_1, R_2 \rp$ is achievable. We have:
\begin{align}
\label{Eq:OuterBoundDerivation}
&n \left( R_1 + \bm R_2 \right) = H(W_1) + \bm H(W_2) \nonumber \\
& \overset{(a)}= H(W_1|W_2, S^n, \si) + \bm H(W_2| S^n, \si) \nonumber \\
& \overset{(\mathrm{Fano})}\leq I(W_1;Y_1^n|W_2, S^n, \si) + \bm I(W_2;Y_2^n|S^n,\si) + n \upxi_n \nonumber \\
& = H(Y_1^n|W_2, S^n, \si) - \underbrace{H(Y_1^n|W_1,W_2,S^n,\si)}_{=~0} \nonumber \\
& \quad + \bm H(Y_2^n|S^n,\si) - \bm H(Y_2^n|W_2,S^n,\si) + n \upxi_n \nonumber \\
& \overset{(b)}\leq \bm H(Y_2^n|S^n,\si) + \smallO(n) + n \upxi_n \nonumber \\
& \overset{(c)}\leq \bm \na \left( 1 - \da \right) + \bm (n-\na-\nt) \left( 1 - \db \right) + \smallO(n) + n\upxi_n \nonumber \\
& \overset{(d)}\leq n \bm \left( 1 - \bar{\delta} \right) + \smallO(n) + n\upxi_n,
\end{align}
where $\upxi_n \rightarrow 0$ as $n \rightarrow \infty$; $(a)$ follows from the independence of messages; $(b)$ follows from Lemma~\ref{Lemma:Leakage_BIC_No} below, which captures the interplay between the varying forward erasure probabilities and the delayed feedback; $(c)$ is true since the entropy of a binary random variable is at most $1$ (or $\log_2(q)$ for packets in $\mathbb{F}_q$) and the forward erasure probabilities are in two non-transient modes and one non-transient mode whose contribution is merged with the $\smallO(n)$ term; and $(d)$ follows form~\eqref{Eq:averagedelta}. Dividing both sides by $n$ and let $n \rightarrow \infty$, we get \eqref{Eq:ConverseBound}.

\begin{lemma}
\label{Lemma:Leakage_BIC_No}
For the two-user \name~with delayed channel state feedback as described in Section~\ref{Section:Problem_BIC} and $\bm$ is given in \eqref{Eq:BetaMax}, we have:
\begin{align}
\label{eq:lemma}
H\left( Y_1^n | W_2, S^n, \si \right) - \bm  H\left( Y_2^n | W_2, S^n, \si \right) \leq \smallO(n),
\end{align}
for any encoding function satisfying~\eqref{eq_enc_function}. 
\end{lemma}


\begin{proof}
We have:
{\small
\begin{align}
\label{Eq:LemmaProof1}
&H\left( Y_2^n | W_2, S^n, \si \right) = \sum_{t=1}^{\na}{H\left( Y_2[t] | Y_2^{t-1},W_2, S^n, \si \right)} \nonumber \\
&~+ \sum_{t=\na+1}^{\na+\nt}{H\left( Y_2[t] | Y_2^{t-1},W_2, S^n, \si \right)} + \sum_{t=\na+\nt+1}^{n}{H\left( Y_2[t] | Y_2^{t-1},W_2, S^n, \si \right)} \nonumber \\
&= \sum_{t=1}^{\na}{(1-\da) H\left( X[t] | Y_2^{t-1},W_2, S[t] = 1, S^{t-1}, S^{t+1:n}, \si \right)} \nonumber \\
&+ \sum_{t=\na+1}^{\na+\nt}{(1-\dt) H\left( X[t] | Y_2^{t-1},W_2, S[t] = 1, S^{t-1}, S^{t+1:n}, \si \right)} \nonumber \\
&+ \sum_{t=\na+\nt+1}^{n}{(1-\db) H\left( X[t] | Y_2^{t-1},W_2, S[t] = 1, S^{t-1}, S^{t+1:n}, \si \right)} \nonumber \\
&\overset{(a)}= \sum_{t=1}^{\na}{(1-\da) H\left( X[t] | Y_2^{t-1},W_2, S^{n}, \si \right)} \nonumber \\
&+ \sum_{t=\na+1}^{\na+\nt}{(1-\dt) H\left( X[t] | Y_2^{t-1},W_2, S^{n}, \si \right)} \nonumber \\
&+ \sum_{t=\na+\nt+1}^{n}{(1-\db) H\left( X[t] | Y_2^{t-1},W_2, S^{n}, \si \right)} \nonumber \\
&\geq \sum_{t=1}^{\na}{(1-\da) H\left( X[t] | Y_1^{t-1}, Y_2^{t-1},W_2, S^{n}, \si \right)} \nonumber \\
&+ \sum_{t=\na+1}^{\na+\nt}{(1-\dt) H\left( X[t] | Y_1^{t-1}, Y_2^{t-1},W_2, S^{n}, \si \right)} \nonumber \\
&+ \sum_{t=\na+\nt+1}^{n}{(1-\db) H\left( X[t] | Y_1^{t-1}, Y_2^{t-1},W_2, S^{n}, \si \right)} \nonumber \\
&= \sum_{t=1}^{\na}{\frac{(1-\da)}{(1-\da^2)} H\left( Y_1[t], Y_2[t] | Y_1^{t-1}, Y_2^{t-1},W_2, S^{n}, \si \right)} \nonumber \\
&+ \sum_{t=\na+1}^{\na+\nt}{\frac{(1-\dt)}{(1-\dt^2)} H\left( Y_1[t], Y_2[t] | Y_1^{t-1}, Y_2^{t-1},W_2, S^{n}, \si \right)} \nonumber \\
&+ \sum_{t=\na+\nt+1}^{n}{\frac{(1-\db)}{(1-\db^2)} H\left( Y_1[t], Y_2[t] | Y_1^{t-1}, Y_2^{t-1},W_2, S^{n}, \si \right)} \nonumber \\
&\overset{(b)}\geq \min\{ 1/\ba, 1/\bb \} H\left( Y_1^n, Y_2^n | W_2, S^{n}, \si \right) + \smallO(n) \nonumber \\
&\overset{(c)}\geq 1/\bm H\left( Y_1^n | W_2, S^{n}, \si \right) + \smallO(n),
\end{align}
}

\noindent where $(a)$ holds since $X[t]$ is independent of $S[t]$ (with the conditioned terms); $(b)$ holds since the omitted term has up to two parts, first part is the product of a discrete entropy term with 
\begin{align}
    \max\{ 1/\ba, 1/\bb \} - \min\{ 1/\beta_A, 1/\beta_B \},
\end{align}
which are both non-negative, and the second part is the product of a summation of discrete entropy terms from $(\na+1)$ to $(\na+\nt)$ with
\begin{align}
    \left( 1/\bt - \max\{ 1/\ba, 1/\bb \} \right)^+ 
\end{align}
which is $\smallO(n)$ due to the length of the interval where $(\cdot)^+ = \max(\cdot,0)$; and $(c)$ follows from the non-negativity of the discrete entropy function and \eqref{Eq:BetaMax}. Then, \eqref{Eq:LemmaProof1} immediately implies \eqref{eq:lemma}.
\end{proof}

\noindent \underline{$\mathbf{\mathcal{C}_{2}}$:} The bounds on the individual rates are already presented. Below, we derive the following outer-bound and the other bound is immediately derived by interchanging user IDs:
\begin{align}
\label{Eq:OuterBoundKappa}
R_1 + \bmin R_2 \leq \bmin \left( 1 - \bar{\delta} \right) + \kappa,
\end{align}
for 
\begin{align}
\kappa = \mathbbm{1}_{\max\{\da,\db\} = \da} \eta (1-\da^2) + \mathbbm{1}_{\max\{\da,\db\} = \db} (1 - \eta)  (1 - \db^2). 
\end{align}

The proof for this bound does overlap to some extent with what has already been presented for \eqref{Eq:ConverseBound}. Thus, we try to avoid repetition in deriving some of the inequalities. Suppose rate-tuple $\lp R_1, R_2 \rp$ is achievable. We have:
\begin{align}
&n \left( R_1 + \bmin R_2 \right) = H(W_1) + \bmin H(W_2) \nonumber \\
& \overset{\eqref{Eq:OuterBoundDerivation}}\leq H(Y_1^n|W_2, S^n, \si) + \bmin H(Y_2^n|S^n,\si) - \bmin H(Y_2^n|W_2,S^n,\si) + n \upxi_n \nonumber \\
& \overset{(a)}\leq \bmin H(Y_2^n|S^n,\si) + \kappa_n + \smallO(n) + n \upxi_n \nonumber \\
& \leq n \bmin \left( 1 - \bar{\delta} \right) + \kappa_n + \smallO(n) + n\upxi_n,
\end{align}
where any omitted derivation is very similar if not identical to \eqref{Eq:OuterBoundDerivation} and $(a)$ follows from Lemma~\ref{Lemma:Leakage_BIC_kappa} below, which also shows that $\lim_{n \rightarrow \infty}\kappa_n/n = \kappa$. Thus, diving both sides by $n$, and taking the limit $n \rightarrow \infty$, we obtain \eqref{Eq:OuterBoundKappa}.

\begin{lemma}
\label{Lemma:Leakage_BIC_kappa}
For the two-user \name~with delayed channel state feedback as described in Section~\ref{Section:Problem_BIC} and $\bm$ is given in \eqref{Eq:BetaMax}, we have:
\begin{align}
\label{eq:lemma-price}
H\left( Y_1^n | W_2, S^n, \si \right) - \bmin  H\left( Y_2^n | W_2, S^n, \si \right) \leq \kappa_n + \smallO(n),
\end{align}
for any encoding function satisfying~\eqref{eq_enc_function}, where
\begin{align}
\label{Eq:Kappan}
\kappa_n \leq \mathbbm{1}_{\max\{\da,\db\} = \da} \na/\ba (1-\da^2) + \mathbbm{1}_{\max\{\da,\db\} = \db} (n - \na - \nt)/\bb  (1 - \db^2).
\end{align}
\end{lemma}

\begin{proof}
We omit any step that is very similar or identical to the proof of Lemma~\ref{Lemma:Leakage_BIC_No}. We have:
{\small
\begin{align}
\label{Eq:LemmaProofKappa}
&H\left( Y_2^n | W_2, S^n, \si \right) \nonumber \\
&\overset{\eqref{Eq:LemmaProof1}}\geq \sum_{t=1}^{\na}{\frac{(1-\da)}{(1-\da^2)} H\left( Y_1[t], Y_2[t] | Y_1^{t-1}, Y_2^{t-1},W_2, S^{n}, \si \right)} \nonumber \\
&+ \sum_{t=\na+1}^{\na+\nt}{\frac{(1-\dt)}{(1-\dt^2)} H\left( Y_1[t], Y_2[t] | Y_1^{t-1}, Y_2^{t-1},W_2, S^{n}, \si \right)} \nonumber \\
&+ \sum_{t=\na+\nt+1}^{n}{\frac{(1-\db)}{(1-\db^2)} H\left( Y_1[t], Y_2[t] | Y_1^{t-1}, Y_2^{t-1},W_2, S^{n}, \si \right)} \nonumber \\
&\overset{(a)}\geq \max\{ 1/\ba, 1/\bb \} H\left( Y_1^n, Y_2^n | W_2, S^{n}, \si \right) + \kappa_n +  \smallO(n) \nonumber \\
&\overset{(b)}\geq 1/\bmin H\left( Y_1^n | W_2, S^{n}, \si \right) + \kappa_n + \smallO(n),
\end{align}
}

\noindent where for $(a)$, the $\smallO(n)$ is similar to what we had in \eqref{Eq:LemmaProof1}, which captures the merger of the non-transient mode between $\na$ and $\na+\nt$, but $\kappa_n$ needs more careful attention. Suppose $1/\ba \geq 1/\bb$, \emph{i.e.}, $\da \leq \db$, then $\kappa_n$ would be:
\begin{align}
(n - \na - \nt) \sum_{t=\na+\nt+1}^{n}{\frac{1}{\bb} H\left( Y_1[t], Y_2[t] | Y_1^{t-1}, Y_2^{t-1},W_2, S^{n}, \si \right)},
\end{align}
but if $1/\ba < 1/\bb$, \emph{i.e.}, $\da > \db$, then $\kappa_n$ would be:
\begin{align}
\na \sum_{t=1}^{\na}{\frac{1}{\ba} H\left( Y_1[t], Y_2[t] | Y_1^{t-1}, Y_2^{t-1},W_2, S^{n}, \si \right)},
\end{align}
which immediately implies \eqref{Eq:Kappan}, and it is a straightforward task to verify $\lim_{n \rightarrow \infty}\kappa_n/n = \kappa$, and $(b)$ follows from the non-negativity of the discrete entropy function and \eqref{Eq:BetaMax}.
\end{proof}

\noindent \underline{\bf Overall outer-bound region:} So far, we have shown that each one of the three regions individually serves as an outer-bound on the capacity region of our problem. Thus, we have:
\begin{equation}
\mathcal{C} \subseteq \mathcal{C}_{1} \cap \mathcal{C}_{2} \cap \mathcal{C}_{3},
\end{equation}
which completes the proof of Theorem~\ref{THM:Bimodal-Outer}.

\section{Achievability of the outer-bounds in Theorem~\ref{THM:Bimodal-Outer}}
\label{Section:Achievability-bimodal}

In this section, we provide the proof for Theorem~\ref{THM:Bimodal-Inner}. But first, we revisit our preliminary work as it paves the way for our proof. In~\cite{Bimodal-ISIT}, we presented a point in the parameters space where $0< \da, \db < 1$ and $0 < \eta < 1$ such that the outer-bound region of Theorem~\ref{THM:Bimodal-Outer} was achievable. However, in this work, we extend the set of parameters for which the outer-bounds are achievable, present more details about the other regime, and discuss potential alternative strategies.

The key finding in this part are two-fold. First, we demonstrate the benefit of inter-modal coding in \BEC s with varying statistics and feedback. In essence, since the transition point, $\na$, is known a priori, the transmitter plans its strategy accordingly to benefit from the multicast aspect of wireless networks as much as possible.
Second, is the fact that the transmitter relies only on the knowledge related to non-transient modes.
For example, the transmitter does not need to know $\nt$ or $\dt$ to achieve the results in this section.  
We then discuss when the inner and outer bounds diverge, and investigate potential future steps.

\subsection{Example presented in~\cite{Bimodal-ISIT}}

\noindent \underline{\bf Setup:} Consider a \name~with channel state feedback having the following parameters: $\da = 0.75$, $\db = 0$, $n_A = \lfloor 32/35n \rfloor$.
In this particular example, the first mode has a high erasure probability, while in the second mode, both links are always on and no erasure occurs. 
Further, the lengths of the modes are chosen carefully as it becomes clear shortly. 

\noindent \underline{\bf Benchmarks:} The first benchmark is the one that ignores the feedback and using erasure codes achieves a sum-rate of $0.25$ in the first mode and $1$ in the second mode for an approximate average sum-rate of $0.31$.
Perhaps the more relevant benchmark is to treat the two modes separately, \emph{i.e.}, intra-modal coding. 
In other words, we can treat the first mode as a \BEC~with delayed channel state feedback having an optimal sum-rate of $7/22 \approx 0.318$ using the well-known three-phase network coding strategy~\cite{georgiadis2009broadcast}, and the second mode as a BEC having an optimal sum-rate of $1$. 
We note that the capacity region of the uni-modal problem can be obtained from our results when $\eta \in \{0 , 1\}$. 
The communication strategy in each mode is well-known and is performed in three separate phases~\cite{georgiadis2009broadcast}: the first two are dedicated to the transmission of uncoded (raw) packets, while the third phase combines the packets delivered to the unintended receiver but not the intended one. 
In doing so, the transmitter creates packets of common interest, which benefit both receivers.
This intra-modal coding strategy results in a weighted average sum-rate of approximately $0.38$.

\noindent \underline{\bf Inter-modal coding:} Instead of the separate treatment of the two modes, for \emph{inter-modal} coding, we start with 
$m = n/5$    
packets\footnote{Without loss of generality, we assume $m \in \mathbb{Z}^+$. Any impact on analysis would vanish as $n \rightarrow \infty$.} for each receiver. 
We then create two phases during the first mode. 
The first phase is dedicated to transmitting the packets intended for user $1$ until \emph{at least} one receiver obtains that packet. 
The transmitter meanwhile keeps track of the status of each transmitted packet: delivered to the intended user, delivered to the unintended user and not the intended user, needs repeating. 
The second phase is similar but dedicated to the packets of user $2$. 
We note that these phases are similar to the first two phases of the typical network coding strategy for \BEC s with feedback.
Each phase will take an average time of\footnote{Here, we limit our analysis to the average values of different random variables. A careful analysis of the communication strategy will entail concentration inequalities similar to~\cite{lin2021capacity}.} 
\begin{align}
    \frac{1}{(1-\da^2)}m = \frac{16}{35}n.
\end{align}
At the end of the second phase (which coincides with the end of the first mode), there will be on average 
\begin{align}
    \frac{\da (1-\da)}{(1-\da^2)}m = \frac{3}{35}n
\end{align}
packets intended for user $i$ available at user $\bar{i}$ but not user $i$ for $\bar{i} = 2 - i$ and $i = 1,2$. Such packets are tracked at the transmitter in a queue represented by $v_{i|\bar{i}}$.

During phase 3, which coincides with the second mode where no erasure occurs, the transmitter at each time sends the summation (XOR in the binary domain) of the packets at the head of queues $v_{1|2}$ and $v_{2|1}$. 
The resulting packets are referred to as multi-cast packets as they assist \emph{both} receivers.
As the links in the second mode are always on, these multi-cast packets are delivered each time and the packets move through the queues one at a time. 
The length of the second mode in this example is carefully chosen to match the (average) number of multi-cast packets that would result from the first mode.
It is straightforward to verify that the receivers can resolve the interference and recover their intended packets at the end of the multi-cast phase. 
Since each user will recover $n/5$ bits successfully over a block-length of $n$, the achievable sum-rate would be $0.4$, which is larger than the weighted average when we treated the two modes separately. 

For the parameters given here, the outer-bounds in Theorem~\ref{THM:Bimodal-Outer} become:
\begin{equation}
\label{Eq:Bimodal-Capacity}
\mathcal{C}^\mathrm{out} \equiv 
\left\{ \begin{array}{ll}
0 \leq \frac{7}{4} R_1 + R_2 \leq \frac{7}{4} \left( 1 - \frac{24}{35} \right), & \\
0 \leq R_1 + \frac{7}{4} R_2 \leq \frac{7}{4} \left( 1 - \frac{24}{35} \right). & 
\end{array} \right.
\end{equation}
These outer-bounds result in a maximum symmetric sum-rate of $0.4$ and as the corner points are easily achievable, the capacity region in this example is known and matches the region described in Theorem~\ref{THM:Bimodal-Outer}.
This example demonstrates the existence of a non-trivial example in which the outer-bound region of Theorem~\ref{THM:Bimodal-Outer} is achievable, and thus, proves Theorem~\ref{THM:Bimodal-Inner}.

\subsection{Proof of Theorem~\ref{THM:Bimodal-Inner}: Improving upon~\cite{Bimodal-ISIT}} 

Here, we show that the outer-bounds are achievable for conditions stated in Theorem~\ref{THM:Bimodal-Inner}, which are presented again below for convenience:
\begin{align}
\label{Eq:ConditionsRepeat}
\da \geq \db, \qquad \text{~and~} \qquad \eta \geq \left( 1 + \frac{\da(1-\da)}{2(1-\db)} \right)^{-1}.
\end{align}
Under these conditions, it is a straightforward task to verify that:
\begin{align}
\mathcal{C}_{1} \cap \mathcal{C}_{2} \cap \mathcal{C}_{3} \equiv \mathcal{C}_{1},
\end{align}
where $\mathcal{C}_{1}$ is given in \eqref{Eq:Region-Outer-da}. Thus, to prove the outer-bounds are achievable, it suffices to prove the achievability of the symmetric sum-rate point obtained from $\mathcal{C}_{1}$ given by:
\begin{align}
\label{Eq:taregtPoint}
R_1 = R_2 = \frac{\bm \left( 1 - \bar{\delta} \right)}{1+\bm} = \frac{(1+\da) ( 1 - \bar{\delta} )}{\left( 2+\da \right)}
\end{align}

A quick check reveals that the example presented above obviously satisfies the first condition and the suggested $\eta$ of $32/35$ satisfies the second condition with equality. Essentially, the right-hand side of the second condition in \eqref{Eq:Conditions} is the point at which the entire mode with lower erasure probability (\emph{i.e.}, the second mode since we assume $\da \geq \db$) becomes fully consumed with the multi-cast packets from the first mode. 

The communication strategy is as follows. Suppose we start with $m$ packets (bits) for each user. Sending the raw packets until at least one user obtains the packet takes on average:
\begin{align}
\frac{m}{(1-\da^2)}.
\end{align}
Thus, the average total time dedicated to the communication of raw packets is given by 
\begin{align}
\label{Eq:DefineAlpha}
\frac{2m}{(1-\da^2)} \overset{\triangle}= n\alpha,
\end{align}
and we restrict $\alpha \leq \eta$.

\begin{remark}[Using statistical average values]
Focusing on the average lengths of different operations enables us to focus on the key aspects of the strategy and provide a clean, intuitive solution. 
In reality, each average term would be accompanied some perturbation. 
However, as it has been done in various recent results~\cite{AlirezaBFICDelayed,IFB-Journal}, a careful use of concentration-type results would show those perturbation terms do not affect the final results. 
For instance, instead of using \eqref{Eq:DefineAlpha} for the first two phases, one would use: 
\begin{align}
n \alpha + \mathcal{O}(n^{2/3}),
\end{align}
to ensure completion of the two phases with probability approaching $1$.
Then, it is a straightforward task to show $\mathcal{O}(n^{2/3})$ terms do not affect the final results as $n \rightarrow \infty$.
\end{remark}

After the transmission of raw packets and as discussed earlier, we denote by virtual queue $v_{i|\bar{i}}$ the packets intended for user $i$ available at user $\bar{i}$ but not user $i$ for $\bar{i} = 2 - i$ and $i = 1,2$. In the following phase, the transmitter will deliver the packets in $v_{1|2} \oplus v_{2|1}$, which denotes the pair-wise addition of the packets (bits) in each virtual queue. The average value of $\left| v_{1|2} \oplus v_{2|1} \right|$ is given by:
\begin{align}
\frac{m\da(1-\da)}{(1-\da^2)}.
\end{align}

The multi-cast phase that follows the transmission of raw packets will consume the remainder $n(\eta - \alpha)$ of the first non-transient mode, the transient mode\footnote{The transmitter may treat the transient mode as the second non-transient mode and the impact would vanish as $n \rightarrow \infty$.}, and the entire length of the second non-transient mode. 
Note that once all multi-cast packets are delivered to both users, all desired packets can be recovered at each user.
Thus, we define the effective delivery rate of the multi-cast packets as:
\begin{align}
R_{\mathsf{eff}} &= \lim_{n \rightarrow \infty}\frac{n(1-\eta-\nt/n)(1-\db) + \nt (1-\dt) + n(\eta-\alpha)(1-\da)}{n(1-\alpha)} \nonumber \\
& = \frac{(1-\eta)(1-\db) + (\eta-\alpha)(1-\da)}{(1-\alpha)}.
\end{align}
Thus, the total communication time normalized by the blocklength can be written as:
\begin{align}
    \alpha + \frac{\left| v_{1|2} \oplus v_{2|1} \right|}{nR_{\mathsf{eff}}} \overset{\eqref{Eq:DefineAlpha}}= \alpha + \frac{\alpha \da (1-\da)}{2 R_{\mathsf{eff}}},
\end{align}
which we wish to be equal to $1$. In other words, we will show that as long as the following quadratic equation in $\alpha$ has a non-negative solution less than or equal to $\eta$, we can achieve the desired rates:
\begin{align}
    \alpha + \frac{\alpha \da (1-\da)}{2 R_{\mathsf{eff}}} = 1.
\end{align}
It is then a straightforward practice in basic calculus to show such a solution exists as long as:
\begin{align}
    \eta \geq \left( 1 + \frac{\da(1-\da)}{2(1-\db)} \right)^{-1},
\end{align}
since one solution is: 
\begin{align}
\label{Eq:SolutionAlpha}
    \alpha=\frac{2(1-\bar{\delta})}{(2+\da)(1-\da)},
\end{align} 
and the other is the impossible value of $\alpha = 1$. Thus, putting the value of $\alpha$ from \eqref{Eq:SolutionAlpha} into \eqref{Eq:DefineAlpha}, we get the following number of packets that can be successfully delivered to each user over a blocklength of $n$:
\begin{align}
    m = \frac{(1-\da^2)\alpha n}{2} \overset{\eqref{Eq:SolutionAlpha}}=\frac{(1+\da)(1-\bar{\delta})n}{(2+\da)},
\end{align}
which immediately implies the achievability of the rates in \eqref{Eq:taregtPoint}. This completes the achievability proof of the maximum sum-rate point given in \eqref{Eq:taregtPoint} and thus, Theorem~\ref{THM:Bimodal-Inner}.

\subsection{Challenges} 

\begin{figure}[!ht]
\centering
\includegraphics[width = .5\columnwidth]{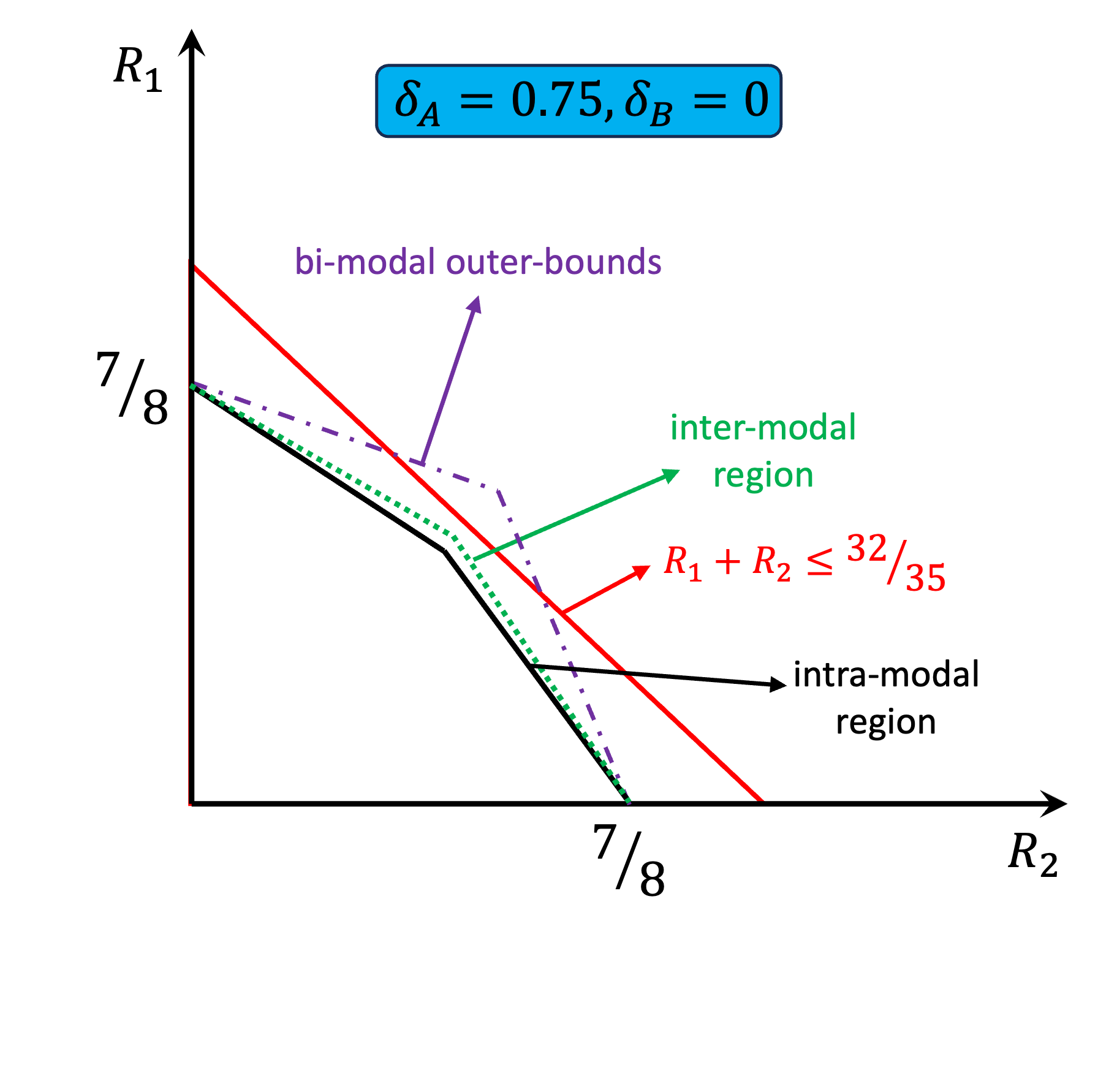}
\caption{An example with $\delta_A = 0.75$, $\delta_B = 0$, and $n_A = \lfloor n/6 \rfloor$ for which the inner and outer bounds deviate. Inter-modal coding nonetheless outperforms intra-modal coding.\label{Fig:BiModal-Challenge}}
\end{figure}

The following example is borrowed from our preliminary work~\cite{Bimodal-ISIT}: keep $\delta_A = 0.75$ and $\delta_B = 0$, but set $n_A = \lfloor n/6 \rfloor$ resulting in $\bar{\delta} = 1/8$. 
For these parameters, the instantaneous feedback outer-bound becomes active and maximum sum-rate is bounded by: 
\begin{align}
    R_1 + R_2 \leq \frac{32}{35} \approx 0.91.
\end{align}
Following a similar approach to the previous example, meaning that using the first mode to send new packets intended for each user in two phases and then using part of the second mode for delivering multi-cast packets, leaves most of the second mode untouched, which is then amended by delivering fresh packets at a sum-rate of $1$. 
This strategy results in a sum-rate of approximately $0.89$, which falls short of the outer-bound as shown in Figure~\ref{Fig:BiModal-Challenge} (this figure is borrowed from~\cite{Bimodal-ISIT}).
Nonetheless, this achievable rate is above the weighted average sum-rate of separate treatment of the two non-transient modes, which gives us an achievable rate of $0.875$.

The proposed achievability strategy is based on the intuition that it is beneficial to transmit the multicast packets during the mode with lower erasure probability. 
However, it is not immediately obvious whether this strategy is necessarily optimal.
One may instead dedicate some fraction of each non-transient mode to the transmission of uncoded packets and the remaining to the multicast packets, and then try and optimize over the choice of these fractions such that the overall achievable rate is maximized.
We numerically evaluated such a strategy and verified that the intuition indeed holds and there is no additional gain with this new idea.

\subsection{Further discussion}

\noindent \underline{\bf Relative channel strength in non-transient modes:}
In the examples we have discussed so far, we assumed the erasure probability during the first non-transient mode is larger than that of the second non-transient mode, \emph{i.e.} $\da > \db$. 
This assumption makes it easier to start with uncoded packets and push the multicast phase of communication to the second non-transient mode.
If instead the first non-transient mode has stronger channels (lower erasure probability), we could potentially achieve the same rates taking advantage of the scheme presented in~\cite{yang2012degrees}, known as reverse Maddah-Ali-Tse scheme.
This scheme starts by sending combinations of uncoded packets and effectively implements similar ideas in reverse.

\noindent \underline{\bf Beyond linear coding:}
The proposed achievability strategy falls under the category of linear coding.
In the context of multi-user erasure channels, it was shown in~\cite{AlirezaNoCSIT} that a non-linear scheme could outperform linear schemes although in the context of erasure interference channels. 
It is then interesting to see as a future direction if any non-linear code could improve the achievable region presented in this work.


\section{Conclusion}
\label{Section:Conclusion_BiModal}

To model MA/RIS-aided wireless networks in higher frequency bands, we introduced the non-ergodic \name~with transient and non-transient modes. 
We provided a new set of outer-bounds as well as an inter-modal coding strategy, which match under non-trivial conditions. 
We also discussed the challenges in other regimes and discussed potential future directions for improving the inner and the outer bounds.
Other future follow ups may include the extension of the results to a more generalized model for feedback signaling, such as intermittent and rate-limited, as well as networks with distributed transmitters such as the erasure interference channel.




\bibliographystyle{ieeetr}
\bibliography{bib_FBBudget}

\end{document}